\documentclass[aps,pra,reprint,superscriptaddress]{revtex4-2}
\usepackage[english]{babel}
\usepackage[utf8]{inputenc}
\usepackage{expl3}
\usepackage{amsmath}
\usepackage{amsfonts}
\usepackage{amssymb}
\usepackage{amsthm}
\usepackage{microtype}
\usepackage{graphicx}
\usepackage{xcolor}
\usepackage{hyperref}
\usepackage{cleveref}
\usepackage{mathtools}
\usepackage{physics}
\usepackage{qcircuit}
\usepackage{widebar}

\let\neg\widebar
\makeatother
\def\@forall@skip{1ex minus .3ex}
\def\forall{\ifmmode{\hskip\@forall@skip\text{for all}\hskip\@forall@skip}%
            \else{ for all }\fi}

\usepackage{amsthm}
\newtheorem{theorem}{Theorem}
\newtheorem{lemma}{Lemma}
\newtheorem{corollary}{Corollary}

\DeclareDocumentEnvironment{comment}{+b}{}{}
\DeclareDocumentEnvironment{sidenote}{}{\bgroup\color{red}}{\egroup}

\ExplSyntaxOn
    \DeclareMathOperator*{\InverseErf}{erf^{-1}}
    \let\OriginalErf\erf % Arxiv does not like NewDocumentCopy
    \cs_new:Nn \erf:N {
        \cs_if_eq:NNTF \inv #1
        { \InverseErf }
        { \OriginalErf#1 }
    }
    \RenewDocumentCommand{\erf}{}{ \erf:N }
\ExplSyntaxOff

\DeclareDocumentCommand\derivative{ s o m g d() }
{ % Total derivative
    \IfBooleanTF{#1}
    {\let\fractype\flatfrac}
    {\let\fractype\frac}
    \IfNoValueTF{#4}
    {
        \IfNoValueTF{#5}
        {\fractype{\diffd \IfNoValueTF{#2}{}{^{#2}}}{\diffd #3\IfNoValueTF{#2}{}{^{#2}}}}
        {\fractype{\diffd \IfNoValueTF{#2}{}{^{#2}}}{\diffd #3\IfNoValueTF{#2}{}{^{#2}}} \argopen(#5\argclose)}
    }
    {\fractype{\diffd \IfNoValueTF{#2}{}{^{#2}} #3}{\diffd #4\IfNoValueTF{#2}{}{^{#2}}}\IfValueT{#5}{(#5)}}
}

\def\et al.{\textit{et al.}}

\NewDocumentCommand{\R}{}{\mathbb{R}}
\NewDocumentCommand{\Reals}{}{\mathbb{R}}
\NewDocumentCommand{\C}{}{\mathbb{C}}
\NewDocumentCommand{\dg}{}{\dagger}
\NewDocumentCommand{\inv}{}{^{-1}}
\NewDocumentCommand{\zo}{s}{\IfBooleanTF{#1}{[0,1]}{\{0,1\}}}

\NewDocumentCommand{\HammingWeight}{m}{\abs{#1}}
\NewDocumentCommand{\SuchThat}{}{\mathrel{\colon}}
\NewDocumentCommand{\Where}{}{\mathrel{,}}
\NewDocumentCommand{\Algorithm}{m}{\ensuremath{\mathcal{#1}}}
\NewDocumentCommand{\bigO}{s}{\IfBooleanTF{#1}{{\mathcal{O}}\qty}{{\mathcal{O}}}}
\NewDocumentCommand{\tbigO}{s}{\IfBooleanTF{#1}{\tilde{\mathcal{O}}\qty}{\tilde{\mathcal{O}}}}
\NewDocumentCommand{\Floor}{sm}{%
    \IfBooleanTF{#1}{\left\lfloor}{\lfloor}%
    {#2}%
    \IfBooleanTF{#1}{\right\rfloor}{\rfloor}}
\NewDocumentCommand{\Sslash}{}{\mskip 0mu plus 2mu/\mkern-6mu/\mkern1mu\mskip 0mu plus 2mu}

\NewDocumentCommand{\Please}{m}{\overset{!}{#1}}

\DeclareMathOperator{\Identity}{I}
\DeclareMathOperator{\Threshold}{Threshold}

\begin{document}

\title{Making the cut: two methods for breaking down a quantum algorithm}

\author{Miguel \surname{Mur\c{c}a}}
\thanks{These authors contributed equally to this work}
\email[corresponding author address: ]{miguel.murca@tecnico.ulisboa.pt}
\affiliation{Instituto Superior T\'{e}cnico, Universidade de Lisboa, Portugal}
\affiliation{Centro de Física e Engenharia de Materiais Avançados (CeFEMA), Physics of Information and Quantum Technologies Group, Portugal}
\affiliation{PQI -- Portuguese Quantum Institute, Portugal}

\author{Duarte \surname{Magano}}
\thanks{These authors contributed equally to this work}
\email[corresponding author address: ]{miguel.murca@tecnico.ulisboa.pt}
\affiliation{Instituto Superior T\'{e}cnico, Universidade de Lisboa, Portugal}
\affiliation{Instituto de Telecomunica\c{c}\~{o}es, Lisboa, Portugal}

\author{Yasser Omar}
\affiliation{Instituto Superior T\'{e}cnico, Universidade de Lisboa, Portugal}
\affiliation{Centro de Física e Engenharia de Materiais Avançados (CeFEMA), Physics of Information and Quantum Technologies Group, Portugal}
\affiliation{PQI -- Portuguese Quantum Institute, Portugal}

\date{\today}

\begin{abstract}
    Despite the promise that fault-tolerant quantum computers can efficiently
    solve classically intractable problems, it remains a major challenge to
    find quantum algorithms that may reach computational advantage in the
    present era of noisy, small-scale quantum hardware. Thus, there is
    substantial ongoing effort to create new quantum algorithms (or adapt
    existing ones) to accommodate depth and space restrictions. By adopting a
    hybrid query perspective, we identify and characterize two methods of
    ``breaking down'' quantum algorithms into rounds of lower (query) depth,
    designating these approaches as ``parallelization'' and ``interpolation''.
    To the best of our knowledge, these had not been explicitly identified and
    compared side-by-side, although one can find instances of them in the
    literature. We apply them to two problems with known quantum speedup:
    calculating the $k$-threshold function and computing a NAND tree. We show
    that for the first problem parallelization offers the best performance,
    while for the second interpolation is the better choice. This illustrates
    that no approach is strictly better than the other, and so that there is
    more than one good way to break down a quantum algorithm into a hybrid
    quantum-classical algorithm.
\end{abstract}

\maketitle

\section{Introduction}

Algorithms that combine classical processing with limited quantum computational
resources hold an attractive promise: to provide computational advantage over
completely classical computation, while remaining compatible with the
technological landscape of quantum computing. The appeal of this kind of
algorithm is well reflected in some of the key modern proposals for quantum
advantage, usually based on variational principles \cite{Bharti2022}. Prominent
examples include the Quantum Approximate Optimization Algorithm
\cite{Farhi2014}, the Variational Quantum Eigensolver
\cite{VQE,McClean2016,Wecker,Kandala2017}, and some versions of Quantum Machine
Learning
\cite{farhi2018classification,Benedetti2019,chen2020variational,Banchi2022,Buffoni2020}.
All of these attempt to exploit circuits of limited coherence to obtain
computational advantage. However, variational approaches often cannot offer
theoretical performance guarantees, as discussed in
refs.~\cite{McClean2018,Anschuetz2022}.

Consider instead a setting where we are given a quantum algorithm with
guaranteed advantage for a certain computational problem, but the available
hardware is too noisy to execute the algorithm with a reasonable fidelity. We
would need to limit the circuit depths to values much shorter than the ones
prescribed by the original algorithm to prevent errors from dominating the
calculations. Is it still possible to guarantee some quantum advantage? We may
phrase this question more precisely. Say we are faced with a computational
problem $f$ that can be solved by a classical computer in time $C(f)$, and we
know a quantum algorithm that solves $f$ with complexity $Q(f)$ ($Q(f) < C(f)$)
by running quantum circuits of depth $D$; what is the best we can do if we are
only permitted to run quantum circuits up to a depth $D'$ smaller than $D$? The
expectation is that best strategy yields an algorithm with a complexity between
$C(f)$ and $Q(f)$. We believe that understanding this question may contribute
to finding practical but provable advantage in near-term quantum computers.

For oracular problems, the notion of limited coherence is captured by the
hybrid query (or decision tree) complexity $Q(f; D)$, introduced by Sun and
Zheng~\cite{Sun2019}. 
In this setting, only the input accesses (or queries)
contribute to the complexity count, while the intermediate computations are
free. 
$Q(f; D)$ is defined as the minimum number of queries required to solve $f$
when limited to running quantum circuits of depth $D$. That is, we can only
perform $D$ queries before being forced to measure the state of the circuit and
restart it.

It is known that quantum decision trees are strictly more powerful than hybrid
decision trees, which are strictly more powerful than classical decision trees.
Concretely, there is a problem $f$ for which $C(f)$ and $Q(f, \bigO(1))$ are
(super-)exponentially separated \cite{AaronsonAmbainis}, and similarly there is
a problem $f$ for which $Q(f, \bigO(1))$ and $Q(f)$ are exponentially separated
\cite{Sun2019}. There are also problems $f$ that exhibit a continuous trade-off
between speedup and circuit depth \cite{Wang2019},
i.e., $Q(f;D) < Q(f;D+1)$ for every $D$ between $1$ and $Q(f)$.

\begin{figure*}[t]
    \includegraphics[width=\linewidth]{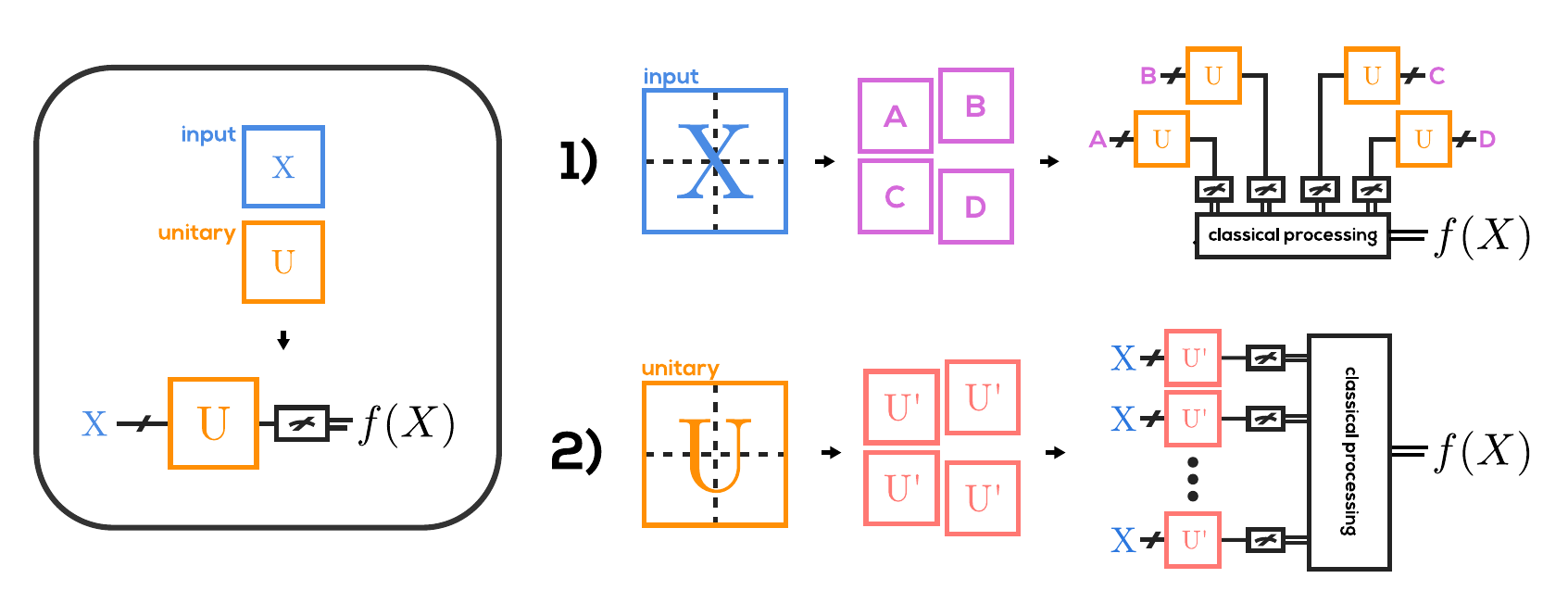}
    \caption{Diagrammatic representation of the procedures of parallelization
    (procedure $1$) and interpolation (procedure $2$), as defined and
    exemplified in this paper. For both procedures, the overall goal is to
    carry out a quantum algorithm (as described by some unitary $U$) for some
    input $X$ to calculate a property $f$ of $X$, as shown in the box to the
    left. However, this unitary may require a prohibitive depth (modelled by us
    as a prohibitive amount of coherent quantum queries). In the case of
    parallelization (procedure $1$), this is dealt with by identifying
    independent, smaller instances of the same problem that can be dealt with
    within the query constraints; in other words, by partitioning the input
    appropriately into sub-problems. Interpolation (procedure $2$) involves,
    instead, considering multiple repetitions of some unitary (or sequence of
    unitaries) that require, individually, less coherent queries, but that
    collectively yield the same information as a single run of $U$. For both of
    these approaches, ``breaking up'' the original algorithm may come at a cost
    of more overall queries -- indeed, we expect that this will be the case for
    most algorithms. We show that no method is strictly better than the other,
    and that the best choice depends on the problem at hand.}
    \label{fig:representation}
\end{figure*}

Despite these landmark results, it is not always obvious how to optimally
``break up'' an algorithm into circuits of smaller sizes. For example, when
given a quantum circuit that is too deep, P\'{e}rez-Salinas \et
al.~\cite{PeresSalinas2022} propose a heuristic algorithm where one performs
intermediate measurements in a parametrized basis, as given by a shallow
variational circuit, optimized to minimize the effect of measuring and
restarting the quantum operation. But, the expressiveness of the shallow circuit
determining the measurement basis and the difficulty of minimizing the cost
function may limit the success of this approach.

In this paper, we identify and discuss two general strategies with theoretical
guarantees to limit the depth of an algorithm to some specified value $D$. We
refer to them as ``parallelization'' and ``interpolation''. Parallelization
applies when a problem can be broken down into a number of smaller, independent
sub-problems, such that the algorithm that solves these sub-problems fits the
permitted circuit depth. In contrast, with interpolation the entire problem is
tackled at each circuit run. It applies whenever there is a trade-off between
the information content of the measurement and the depth of the corresponding
circuit. In these cases, we may compensate the information loss caused by
shortening the circuit depth with repeated runs of the shorter circuit.
Intuitively,  we can say that interpolation methods ``break up the unitary''
that solves the problem instead of breaking up the problem itself. See Figure
\ref{fig:representation} for an illustration of these notions. To the best of
our knowledge, neither of these methods had been explicitly identified and
compared side-by-side, even though several works that fit into these labels can
be found in the literature. For example, parallelization approaches are present
in refs.~\cite{Zalka1999,GroverRadhakrishnan2004,Jeffery2017}, while
refs.~\cite{Wang2019,WangKohJohnsonCao2021,GiurgicaTiron2022,Magano2022}
describe interpolation methods.

We argue that Quantum Singular Value Transformations (QSVTs) \cite{Gilyen2019}
provide a natural framework for thinking about interpolation. In a seminal
work, Gily\'{e}n \et al.~\cite{Gilyen2019} have shown that it is possible to
perform polynomial transformations on the singular values of large matrices
with circuits whose degree is proportional to the degree of the corresponding
polynomial, and that many important quantum algorithms can be described this
way. Taking this perspective, limiting the circuit depth means implementing a
rougher approximation to the target function. As a consequence, each
measurement provides less accurate information about the quantity that we are
trying to estimate. Sometimes, this effect can be compensated with statistical
sampling. That is, we can trade-off circuit depth and number of circuit runs.

We illustrate these methods with two well-known problems: the \(k\)-threshold
function and perfectly balanced NAND trees. These problems are known to exhibit
quantum speed-ups (\cite{Beals2001} and
\cite{Farhi2014,ChildsCleveJordanYongeMallo2009,Ambainis2010}), but, to the
best of our knowledge, neither has been discussed in the context of a
limited-depth computing model. Both problems are amenable to both
parallelization and interpolation. We show that for the \(k\)-threshold
function parallelization offers the best performance, while for evaluating
perfectly balanced NAND trees the interpolation method is the most efficient.
This reinforces the relevance of the distinction between parallelization and
interpolation, and demonstrates that no technique  is \textit{a priori} better
than the other, as the best option depends on the problem at hand.

\section{Preliminaries}

\subsection{Hybrid Query Model \label{sec:querymodel}}

We will be working mostly within the query model of quantum computing. Here we
quickly review the main concepts, referring to Ambainis \cite{AndrisReview2017}
for a more in-depth discussion.

The quantum query complexity model, a generalization of decision tree
complexity \cite{BdW02}, is widely used to study the power of quantum
computers. On one hand, it captures the important features of most quantum
algorithms, including search \cite{Grover1996}, period-finding \cite{Shor1994},
and element distinctness \cite{Ambainis2007}. On the other hand, it is simple
enough to make the proof of lower bounds attainable
\cite{Beals2001,Ambainis2002}.

In the query model, the goal is to compute a Boolean function $f(x_1, \ldots,
x_N)$ of variables \(x_i \in \zo\). The function can be total (defined on
$\zo^N$) or partial (defined on a subset of $\zo^N$). We only get information
about the input variables by querying a black-box quantum operator $O$ acting
as
\begin{equation}
    O \ket{i}\ket{b} = \ket{i} \ket{b \oplus x_i}
\end{equation}
for every $b\in \zo$ and $i \in \zo^N$. A quantum query algorithm is specified
by a set of input-independent unitaries $U_0, U_1,\ldots, U_T$. The algorithm
consists in performing the transformation
\begin{equation}
    \label{eq:querycircuit}
    U_T \, O \, U_{T-1} \, \ldots U_2 \, O \, U_1 \, O \, U_0 \ket{0}
\end{equation}
and measuring the result, which is then converted into the answer of the problem according to a predefined rule.
In the query model, the algorithm's complexity increases with each query, while the intermediate computations are free.
That is, the complexity of the algorithm corresponding to transformation \eqref{eq:querycircuit} is $T$, independently of how the unitaries $U_i$ are chosen.

We say that a quantum algorithm computes $f$ with bounded error if, for all $x
\in \zo^N$, the answer of the algorithm agrees with $f(x)$ with probability at
least $2/3$, where the probability is over the randomness of the algorithm's
measuring process. The minimum query complexity of any bounded-error algorithm
computing $f$ is the quantum (bounded-error) complexity of $f$, denoted as
$Q(f)$.

The hybrid query model introduced by Sun and Zheng \cite{Sun2019} captures the
idea of restricted-depth computation in an oracular setting. Hybrid algorithms
are in direct correspondence with hybrid decision trees. A hybrid decision tree
is similar to a (classical) decision tree, but the decision at each node is
determined by the output of a quantum  algorithm with query complexity no more
than a value $D$, which we refer to as the depth of the hybrid algorithm. The
hybrid algorithm's answer is the output of the algorithm at the leaf node. More
plainly, a hybrid algorithm works by running and measuring sequences of
circuits like \eqref{eq:querycircuit} with $T \leq D$, using the intermediate
measurements to decide what quantum circuit to run next.

A hybrid algorithm computes $f$ with bounded error if, for all $x \in \zo^N$,
the answer of the algorithm agrees with $f(x)$ with probability at least $2/3$,
where the probability is over the randomness of the internal measurements. The
complexity of a path in a hybrid tree is the sum of the complexities of the
algorithms associated to each node in the path. The complexity of a hybrid
algorithm that computes a function $f$ is the maximal complexity of any path
that connects the root and a leaf, that is, it is the total number of queries
needed to evaluate $f$ in the worst case. The minimum query complexity of any
bounded-error hybrid algorithm computing $f$ is the hybrid (bounded-error)
complexity of $f$, denoted as $Q(f; D)$.

\subsection{Quantum Singular Value Transformations \label{sec:QSVT}}

In this paper we make extensive use of Quantum Singular Value Transformations
(QSVTs) \cite{Gilyen2019,Martyn2021,Lin2022}. As a generalization of the work
on Quantum Signal Processing \cite{LowChuang19}, QSVTs have provided a unifying
description of several algorithms, including  amplitude estimation, quantum
simulation, and quantum methods for linear systems. Recently, Magano and
Mur\c{c}a \cite{Magano2022} have shown that QSVTs also constitute a natural
framework for reasoning about interpolation methods.

By the singular value decomposition theorem, an arbitrary matrix $A$ of rank
$r$ can be written as
\begin{equation}
    A = \sum_{i=1}^r \sigma_i \ket{w_k} \bra{v_k},
\end{equation}
where $\{w_k\}_k$ and $\{v_k\}_k$ are orthogonal sets (known as the left and
right singular values of $A$, respectively) and $\{\sigma_k\}_k$ are positive
real numbers (known as the singular values of $A$). For functions $P: \R
\rightarrow \C$, we call
\begin{equation}
    P^{\text{(SV)}}(A) := \sum_{i=1}^r P(\sigma_i) \ket{w_k} \bra{v_k}
\end{equation}
a singular value transformation of $A$.

When considering performing such transformations on arbitrary matrices with
quantum computers, we are immediately faced with the difficulty that quantum
states evolve according to unitary transformations. The introduction of
\emph{block-encodings} overcomes this apparent limitation
\cite{Chakraborty2018}. Let \(\Pi\) and \(\tilde\Pi\) be orthogonal projectors
and \(U\) be a unitary; we say that \(\Pi\), \(\tilde\Pi\), and $U$ form a
block-encoding of the operator $A$ if
\begin{equation}
    A = \tilde\Pi U \Pi.
\end{equation}

Based on this concept, the main theorem of QSVTs can be phrase as follows.

\begin{theorem}[QSVTs \cite{Gilyen2019}] 
    \label{thm:qsvt}
    Let \(\Pi\), \(\tilde\Pi\), and \(U\) be a block-encoding of a matrix $A$,
    and let $P : [-1,1] \rightarrow [-1,1]$ be a polynomial of degree $d$.
    Then, we can implement a unitary $U_P$ such that \(\Pi\), \(\tilde\Pi\),
    and \(U_P\) form a block-encoding of $P^{\text{(SV)}}(A)$ using $\bigO(d)$
    calls to $U$, $U^{\dg}$ and $\Pi/\tilde\Pi$-controlled-NOT
    operations.\footnote{By $\Pi$-controlled-NOT we mean an operation that acts
    as a NOT gate controlled on a given state being in the image of
    $\Pi/\tilde\Pi$.} 
\end{theorem}

A transformation that will be particularly useful for us is the step (or
Heaviside) function,
\begin{equation}
    \sigma \mapsto 
    \begin{cases}
        1, \text{ if } \sigma \geq \mu \\
        0, \text{ if } \sigma < \mu
    \end{cases},
\end{equation}
for some $\mu \in [-1,1]$. Refs.~\cite{LowPhD,Gilyen2019} show that we can
approximate this transformation up to arbitrary accuracy by a polynomial
approximation of the error function, defined as
\begin{equation}
    \erf(x) := \frac{2}{\sqrt{\pi}} \int_0^x e^{-t^2} \dd{t}.
\end{equation}

The result is stated below.

\begin{theorem}[Polynomial approximation of step function \cite{Gilyen2019}]
\label{thm:stepfunction}
    There is a polynomial $P_{\delta, \eta, \mu}(\lambda):[-1,1] \rightarrow
    [-1,1]$ of degree
    \begin{equation}
        \bigO \left( \frac{1}{\delta} \log\left(\frac{1}{\eta} \right) \right)
    \end{equation}
    satisfying
    \begin{align}
        \vert P_{\delta, \eta, \mu}(\sigma) \vert  \leq \eta, \forall \sigma &\in [-1,\mu - \delta] && \\
        P_{\delta, \eta, \mu}(\sigma) \geq 1 - \eta, \forall \sigma          &\in [\mu+\delta,1] . &&
    \end{align}
\end{theorem}

We will also be interested in performing a step transformation on the modulus
of the singular values (also known as a window function due to the shape of its
plot),
\begin{equation}
    \sigma \rightarrow 
    \begin{cases}
        1, \text{ if } \vert \sigma\vert  \leq \mu \\
        0, \text{ if } \vert \sigma\vert  > \mu
    \end{cases}.
\end{equation}
Noting that $P_{\delta, \eta, -\mu} - P_{\delta, \eta,\mu}$ is a polynomial
with the same degree as $P_{\delta, \eta, \mu}$, we immediately derive the
following.
\begin{corollary}[Polynomial approximation of window function]
\label{thm:bumpfunction}
    There is a polynomial $P'_{\delta, \eta, \mu}(\lambda):[-1,1] \rightarrow
    [-1,1]$ of degree
    \begin{equation}
        \bigO \left( \frac{1}{\delta} \log\left(\frac{1}{\eta} \right) \right)
    \end{equation}
    satisfying
    \begin{align}
        & \vert P'_{\delta, \eta, \mu}(\sigma) \vert  \leq \eta,
        \forall \sigma \in [-1 , -\mu - \delta] \cup [\mu + \delta, 1]\\
        & P'_{\delta, \eta, \mu}(\sigma)  \geq 1 - \eta, \forall \sigma \in [-\mu + \delta, \mu-\delta] .
    \end{align}
\end{corollary}

Combining Theorems \ref{thm:qsvt} and \ref{thm:stepfunction} we find a method
to distinguish the singular values of a block-encoded matrix that are above or
below a given threshold. Similarly, from Theorem \ref{thm:qsvt} and Corollary
\ref{thm:bumpfunction} we can distinguish the singular values of a
block-encoded matrix whose modulus are above or below a given threshold.

\section{Two approaches to restricted-depth computation}

\subsection{Parallelization}

In many cases the problem at hand can be broken down into a number of smaller,
independent sub-problems. As an example, consider the problem of computing the
OR function on $N$ bits. We can partition the domain into $p$ subdomains of
size approximately $N/p$. If for any of those subdomains there is an index $i$
for which $x_i = 1$, then we return $1$; otherwise the answer is $0$. In other
words, the problem is reduced to evaluating $p$ OR function on $N/p$ bits. With
Grover's algorithm \cite{Grover1996} we can evaluate each subdomain with
$\bigO(\sqrt{N/p})$ queries.
In total, this strategy has a query complexity of 
\begin{equation}
    \bigO*(\sqrt{p N}).
\end{equation}
If we are limited to circuits of depth $D$, we set $p = \bigO(N / D^2)$,
finding that
\begin{equation}
    Q(\text{OR}; D) = \bigO*(\frac{N}{D} + \sqrt{N}).
\end{equation}
By Corollary~1.5 of Sun and Zheng \cite{Sun2019}, this is optimal.

We say that the algorithms that employ this kind of strategy -- breaking the
problem into smaller problems that fit the permitted depth -- fall into the
category of parallelization methods. Note that this procedure does not
require multiple quantum processors operating at the same time, even
though it is amenable to it. The important point is that the different
sub-problems considered are independent and may be treated as such. This should
be contrasted with the notion of parallel quantum algorithms as defined by
Jeffery \et al.\ \cite{Jeffery2017}, where a number of queries are realized at
the same time (in parallel), but by a number of quantum registers that may, for
example, be entangled with each other.

Arguably, parallelization as described above is the most natural approach to
``breaking up'' a quantum algorithm into circuits of lower quantum depth.
Examples of parallelization include Zalka \cite{Zalka1999}, containing the OR
function discussed above, Grover and Radhakrishnan
\cite{GroverRadhakrishnan2004}, searching for marked elements over many copies
of a database, and Jeffery \et al. \cite{Jeffery2017}, with the problems of
element distinctness and $k$-sum. Although these references were originally
motivated by the idea of quantum processors acting in parallel, they easily
translate to the discussion of restricted-depth setting, and fit into the
description of the procedure of parallelization we have made above.

\subsection{Interpolation}

Contrary to parallelization, interpolation methods do not distribute the
problem into different sub-problems. Instead, at each run the entire problem is
tackled -- only over several quantum circuit runs. Since the circuit depth is
limited, each circuit measurement can only yield partial information about the
answer to problem; the definitive answer is recovered by repeating the
computation multiple times.

We illustrate this approach with an information-theoretic argument (similar to
that of Wang \et al.~\cite{WangKohJohnsonCao2021}). Say that we have a quantum
routine $\mathcal{A}$ that prepares the state
\begin{equation} \label{eq:AAstate}
    \ket{0^n} 
    \xrightarrow{\mathcal{A}} 
    \sqrt{1-p} \ket{\psi_0} + \sqrt{p} \ket{\psi_1}
\end{equation}
for some unknown $p \in [0,1]$, and assume that we can efficiently distinguish
between $\ket{\psi_0}$ and $\ket{\psi_1}$. The goal is to estimate $p$, noting
that many query problems can be reduced to estimating an amplitude. With
Grover's iterator \cite{Grover1996}, we can prepare the state
\begin{equation} \label{eq:iteratedAAstate}
    \cos\big( (1 + 2k) \theta\big) \ket{\psi_0}
    + 
    \sin\big((1 + 2k) \theta\big) \ket{\psi_1},
\end{equation}
where $\theta = \arcsin(\sqrt{p})$, with $\bigO(k)$ calls to $\mathcal{A}$.
Now suppose that we prepare and measure the state \eqref{eq:iteratedAAstate} in the $\{\ket{\psi_0}, \ket{\psi_1}\}$ basis $l$ times, recording the outcomes.
The Fisher information associated with this experiment is
\begin{align} 
    I(\pi) := & \,  l \sum_{i=0,1} \frac{1}{\mathbb{P}[\ket{\psi_i} \vert \pi] } \left( \frac{\partial}{\partial \pi} \mathbb{P}[\ket{\psi_i} \vert \pi] \right)^2
    \nonumber
    \\ = & \,  \frac{l (1+ 2k)^2}{\pi(1-\pi)},
    \label{eq:Fisherinfo}
\end{align}
where $\mathbb{P}[\ket{\psi_i} \vert \pi]$ is the probability of observing
outcome $\ket{\psi_i}$ in a single trial assuming that $p = \pi$. Expression
\eqref{eq:Fisherinfo} reveals that the measurement is more informative the
larger the value of $k$ (in particular, that it grows quadratically with $k$,
justifying the quadratic speedup of Grover's algorithm).

Refs.~\cite{Wang2019,GiurgicaTiron2022,Magano2022} have suggested different
schemes to harness the enhanced information of deeper circuits. Here we adopt
the perspective put forward by Magano and Mur\c{c}a \cite{Magano2022},
according to which QSVTs constitute a natural framework for interpolation
methods. The idea is to trade off the quality of the polynomial approximation
to the target function by statistical sampling. That is, we can compensate
using polynomials of lower degree (corresponding to lower circuit depths) by
running the quantum circuits more times. The result is a continuous trade-off
between circuit depth and quantum speed-up, without ever needing to identify
independent sub-problems.

In the subsequent sections we demonstrate how QSVTs can be used to interpolate
specific problems.

\section{Sometimes Parallelization is Better: Threshold Function}

Consider the \(k\)-threshold function, a total symmetric boolean function
defined as follows:

\begin{equation}
    \Threshold_k (x_1, \ldots, x_N) = \begin{cases}
        0 & \text{if } \sum_{i=1}^N x_i \leq k \\
        1 & \text{otherwise}
    \end{cases}.
\end{equation}

This function admits a quantum query speed-up: whereas in the classical case
\(\Theta(N)\) queries are required (easily concluded by an adversarial
argument), the quantum query complexity is \(\Theta(\sqrt{N \min(k, N-k)})\)
(as follows from Beals \et al. \cite{Beals2001}), resulting in the
aforementioned quadratic speed-up when \(\min(k, N-k)=\bigO(1)\), and no
speed-up when \(\min(k, N-k) = \Omega(N)\). For simplicity, we assume from now
on that $k \leq N/2$.

We approach the problem from the perspective of QSVTs. This is a departure from
the original proof of Beals \et al.~\cite{Beals2001}, where the problem of
evaluating any totally symmetric Boolean function is reduced to quantum
counting. Arguably, QSVTs permit tackling the $k$-threshold problem more
directly, while also offering a more natural route towards interpolation. We
show in Appendix \ref{sec:symmfunctions} that our approach can also be
generalized to any totally symmetric Boolean function, although in that case
the proof resembles more closely that of Beals \et al.~\cite{Beals2001}.

We start by making the (trivial) observation that $k$-threshold function can be
written as a function of the Hamming weight of the input, which we denote by
$\HammingWeight{x}$. The first step of our algorithm will be to block-encode
$\sqrt{\HammingWeight{x} / N}$ (or, more technically, to block-encode the
$1\times1$ matrix whose only entry is $\sqrt{\HammingWeight{x} / N}$). Then, we
will perform a QSVT on this value to prepare the desired function of
$\HammingWeight{x}$.

Consider the unitary transformation
\begin{equation}
    U = \quad\raisebox{1.2em}{%
        \Qcircuit @C=1.1em @R=.5em {
            & \qw{{}^{n}/} & \gate{H^{\otimes n}}     & \multigate{1}{O_X} & \qw \\
            & \qw          & \qw                    & \ghost{O_X}        & \qw \\
        }
    },
\end{equation}
where \(n = \log_2(N)\) -- assuming, without loss of generality, that \(N\) is
exactly a power of two -- and \(O\) is our query operator (defined in Section
\ref{sec:querymodel}). We have that
\begin{flalign}
    U \ket*{0^{n+1}} & = \sqrt{\frac{1}{N}} \left(\quad\sum_{\mathclap{i \SuchThat x_i = 0}} \ket{i}\ket{0} + \sum_{\mathclap{i \SuchThat x_i = 1}} \ket{i} \ket{1}\;\right) \nonumber \\
    & = \sqrt{1 - \frac{\HammingWeight{X}}{N}} \ket{\phi_0} \ket{0} + \sqrt{\frac{\HammingWeight{X}}{N}} \ket{\phi_1} \ket{1} ,
    \label{eq:threshold_blockencoding}
\end{flalign}
where \(\ket{\phi_0}\) and \(\ket{\phi_1}\) are normalized states.
Choosing
\begin{equation}
    \Pi = \dyad*{0^{n+1}} \, \text{ and } \, \tilde\Pi = \Identity_{2^n} \otimes \dyad{1}
\end{equation}
we find that
\begin{equation}
    \tilde\Pi U \Pi = \sqrt{\frac{\HammingWeight{x}}{N}}.
\end{equation}

That is, $\tilde\Pi$, $\Pi$, and $U$ form a block-encoding of
$\sqrt{\HammingWeight{x} / N}$.

We would like to distinguish between cases where $\sqrt{\HammingWeight{x} / N}$
is smaller or equal to $\sqrt{k / N}$ and those where it is larger than
$\sqrt{k / N}$. From the results on QSVTs (Theorems \ref{thm:qsvt} and
\ref{thm:stepfunction}) we can perform the transformation
\begin{equation}
    \label{eq:threshold_transformation}
    \ket*{0^{n+1}}
    \rightarrow
    P_{\delta, \eta, \mu}\left(\sqrt{\frac{\HammingWeight{x}}{N}} \right) 
    \ket{\phi_1} \ket{1}
    + \ket{\perp_1},
\end{equation}
where $\ket{\perp_1}$ is such that $\tilde\Pi \ket{\perp_1} = 0$, using $\bigO
( (1/\delta) \log (1 / \eta ) )$ calls to $U$. As $U$ only calls the query
operator $O$ once, the operation \eqref{eq:threshold_transformation} only
involves $\bigO ( (1/\delta) \log (1 / \eta ) )$ queries. We choose the
parameters as
\begin{align}
    \eta &= 1/8, \label{eq:eta_def}\\
    \delta &= \frac{1}{2} \left(\sqrt{(k+1)/N} - \sqrt{k/N} \right)
    = \bigO*(\frac{1}{\sqrt{kN}}), \label{eq:par_delta_def} \\
    \mu &= \frac{1}{2} \left(\sqrt{(k+1)/N} + \sqrt{k/N} \right),
\end{align}
in which case the operation \eqref{eq:threshold_transformation} consumes
$\bigO(\sqrt{kN})$ queries. The final step is simply to measure the last qubit
of the resulting state, outputting $0$ if we measure $\ket{0}$ and outputting
$1$ if we measure $\ket{1}$. To verify that this yields the desired answer,
consider the two possible scenarios:
\begin{itemize}
    \item $\Threshold_k (x_1, \ldots, x_N) = 0$. Then, $\sqrt{\HammingWeight{x}
        / N} \leq \sqrt{\HammingWeight{k} / N}$, which means that $P_{\delta,
        \eta, \mu}(\sqrt{\HammingWeight{x} / N} ) \leq \eta = 1/8$. So, the
        probability of measuring the last qubit in state $\ket{1}$ is less than
        $1/3$.
    \item $\Threshold_k (x_1, \ldots, x_N) = 1$. Then, $\sqrt{\HammingWeight{x}
        / N} > \sqrt{\HammingWeight{k} / N}$, which means that $P_{\delta,
        \eta, \mu}(\sqrt{\HammingWeight{x} / N} ) \geq  1 - \eta = 7/8$. So,
        the probability of measuring the last qubit in state $\ket{1}$ is
        greater than $2/3$.
\end{itemize}

If instead $k > N/2$, the algorithm does not change significantly: denote the
logical negation of $x$ by $\neg x$, and note that $\Threshold_k (x_1, \ldots,
x_N) =  1 - \Threshold_k (\neg x_1, \ldots, \neg x_N)$. It follows that we just
need to evaluate the threshold function on $\neg x$, whose Hamming weight is
$\HammingWeight{\neg x} = N - \HammingWeight{x}$. Looking at expression
\eqref{eq:threshold_blockencoding}, we see that $U$ already provides a
block-encoding of the $\sqrt{\HammingWeight{\neg x}/N}$: we just need to
replace $\tilde\Pi$ by $\Identity_{2^n} \otimes \dyad{0}$. Everything else
follows as before.

\paragraph{Interpolation.}

The algorithm that we have just presented can be interpolated using the same
strategy as in Magano and Mur\c{c}a \cite{Magano2022}. Recall from the theory
of QSVTs (Section \ref{sec:QSVT}) that with deeper circuits we can prepare
polynomial transformations of higher degree. Conversely, by limiting the circuit
depths we are forced to implement a rougher approximation to the target
function (in this case, the step function). The idea is to compensate this
effect by performing a larger number of measurements.

Concretely, the trade-off between circuit depth and repetitions of the circuit
can be controlled by the parameter $\eta$, which we had previously fixed to be
$\bigO(1)$ (\textit{cf.} \eqref{eq:eta_def}). Now we choose
\begin{equation}
    \eta = \bigO\left(2^{-\delta D}\right)
\end{equation}
in such a way that the circuit depth associated with the transformation by
$P_{\delta, \eta, \mu}$ is upper bounded by $D$. If we measure the last qubit
of state \eqref{eq:threshold_transformation}, the probability that we see
$\ket{1}$ is
\begin{align}
    & \leq  \eta^2,  \text{ if } \Threshold_k (x) = 0, \text{ or} \\
    & \geq  (1-\eta)^2,  \text{ if } \Threshold_k (x) = 1.
\end{align}

So, the problem is reduced to distinguishing the bias of a Bernoulli
distribution with precision $1 - 2 \eta$. It is well-known that $\Theta(1 /
(1-\eta)^2)$ samples are sufficient (and necessary) to achieve such a precision
with bounded-error probability. That is, we prepare and measure state
\eqref{eq:threshold_transformation}
\begin{equation}
    \bigO*(\frac{1}{(1-\eta)^2}) = \bigO*(\frac{1}{\delta D})
\end{equation} 
times. The total number of queries to $O$ is
\begin{equation}
    \bigO*( \frac{1}{(1-\eta)^2} \times \frac{1}{\delta} \log\left(\frac{1}{\eta} \right) )
    =
    \bigO*(\frac{1}{\delta^2 D}).
\end{equation}

Replacing in the definition of $\delta$ \eqref{eq:par_delta_def}, we conclude that
\begin{equation}
    \label{eq:thresholdcomplexity_interpolation}
    Q(\Threshold_k; D) = \bigO*(\frac{k N}{D} + \sqrt{k N}).
\end{equation}

\paragraph{Parallelization.}

The approach of \cite{Magano2022} was originally developed in the context phase
estimation. In phase estimation the parameter $\phi$ to be estimated is
accessed via a black-box oracle that changes the phase of a particular state by
an angle proportional to $\phi$. In that case, the interpolation is likely
optimal. However, the threshold problem has more structure than phase
estimation. Indeed, we can choose to query only a subset of the input
variables, in which case the block-encoding holds information about the Hamming
weight of that subset of input variables, whereas we cannot choose to query a
``fractional phase''.

It is the parallelization approach that yields the optimal algorithm for
evaluating the threshold function in a restricted-depth setting. To show this,
we follow a procedure similar to that of Grover and Radhakrishnan
\cite{GroverRadhakrishnan2004}. First, we partition the set \(\{1, 2, \ldots,
N\}\) into \(p\) disjoint subsets $V_1, \ldots, V_p$ of size $N/p$ (to simplify
the notation, we assume that $N/p$ is an integer). Then, for each subset
\(V_i\), we prepare the uniform superposition \(\sqrt{p/N} \sum_{j\in V_i}
\ket{j}\ket{0}\) and apply to it the query operator \(O\). The resulting state
is
\begin{equation}
    \sqrt{\frac{p \HammingWeight{x \Sslash V_i}}{N}} \ket{\phi'_1} \ket{1} + \sqrt{1 - \frac{p \HammingWeight{x \Sslash V_i}}{N}} \ket{\phi'_0} \ket{0}
\end{equation}
where \(\ket{\phi'_0}\), \(\ket{\phi'_1}\) are normalized states and \(\vert x
\Sslash V_i\vert  := \vert\{x_j \in x \SuchThat j \in V_i\} \vert\). If we run
the amplitude estimation algorithm of Brassard \et al.~\cite{Brassard2002} for
$D$ steps we get an estimate of the amplitude $\sqrt{p\HammingWeight{x \Sslash
V_i} / N}$ up to precision
\begin{equation}
    \bigO*( \frac{1}{D}  \sqrt{\frac{p \HammingWeight{x \Sslash V_i}}{N}})
    \label{eq:AAprecision}
\end{equation}
with a constant probability. To lower the probability that the algorithm fails
to $1/p$, we repeat the amplitude amplification routine $\bigO(\log p)$ times;
this guarantees a bounded probability that all the amplitude estimations
succeed in returning a precision as in \eqref{eq:AAprecision}. We set
\begin{equation}
    D = 
    \begin{cases}
        \bigO*(\sqrt{\frac{N \log p}{p}}) & \text{if } k \leq p \log p \\
        \bigO*(\frac{\sqrt{N k}}{p})      & \text{if } k \geq p \log p.
    \end{cases}
    \label{eq:depth_def}
\end{equation}

Then, for every subset $V_i$, we are estimating $\HammingWeight{x \Sslash V_i}$
with precision
\begin{equation}
    \epsilon_i = 
    \begin{cases}
        \bigO*(\sqrt{\frac{\HammingWeight{x \Sslash V_i}}{\log p}}) & \text{if } k \leq p \log p \\
        \bigO*(\sqrt{\frac{\HammingWeight{x \Sslash V_i}}{k / p}})  & \text{if } k \geq p \log p.
    \end{cases}
    \label{eq:xprecision}
\end{equation}

We estimate $\HammingWeight{x}$ as the sum of our estimates for
$\HammingWeight{x \Sslash V_i}$. If it exceeds $k$, we output $1$, and
otherwise we output $0$.

The actual behaviour of the algorithm depends on how the $1$-input
variables are distributed among the subsets $V_1, \ldots, V_p$. In the
worst-case scenario, all the ones are concentrated in a single bin. However,
this scenario is extremely unlikely. Raab and Steger's ``balls into bins''
theorem \cite{BallsIntoBins} states that, with probability greater than $2/3$,
\begin{equation} \label{eq:ballsintobins}
    \max_i \HammingWeight{x \Sslash V_i} = 
    \begin{cases}
        \bigO*(\log p)                      & \text{if } \HammingWeight{x} \leq p \log p \\
        \bigO*(\frac{\HammingWeight{x}}{p}) & \text{if } \HammingWeight{x} \geq p \log p .
    \end{cases}
\end{equation}

Using this result, we show in Appendix \ref{sec:threshold_proof} that there is
a choice for the constant factors in \eqref{eq:depth_def} that guarantees that
our estimate for $\HammingWeight{x}$ is larger than $k$ if $\Threshold_k(x) =
1$ and smaller or equal than $k$ if $\Threshold_k(x) = 0$.

Putting everything together, we conclude that
\begin{equation}
    Q(\Threshold_k; D) = \bigO*(\frac{N}{D} \log^2\left(\frac{N}{D}\right) + \sqrt{Nk} \log k).
\end{equation}

Comparing with the upper bound that we derived with the interpolation method
(equation \eqref{eq:thresholdcomplexity_interpolation}), we see that
parallelization offers the best performance. Indeed, for short circuit depths
the complexity of the parallelization method is smaller by a factor of $k$ (up
to logarithmic factors).

\section{Sometimes Interpolation is Better: NAND trees}

We now apply the interpolation and parallelization techniques for the problem
of evaluating a balanced binary NAND formula. This problem has been widely
studied in the literature: Farhi \et al. \cite{FarhiGoldstoneGutmann2007}
proposed a quantum walk algorithm that runs in $\bigO*(N^{1/2})$ time with an
unconventional, continuous-time query model. Later, Childs \et al.
\cite{ChildsCleveJordanYongeMallo2009} understood that this algorithm could be
translated into the discrete query model (as presented in Section
\ref{sec:querymodel}) with just an $\bigO*(N^{o(1)})$ overhead. Finally,
Ambainis \et al. \cite{Ambainis2010} presented an optimal
$\bigO*(N^{1/2})$-time algorithm on the conventional query model. We adapt
their approach to a restricted-depth setting.

Let $\Phi$ be a Boolean function on $N$ inputs $x_1, \ldots, x_N$ expressed
with NAND gates. We treat  each occurrence of a variable separately, in that
$N$ is counting with the variables' multiplicity. Equivalently, we could be
considering a formula expressed in terms of the gate set $\{\text{AND},
\text{OR}, \text{NOT}\}$. The input is accessed via the conventional query
operator $O$ as defined in Section  \ref{sec:querymodel}.

The formula $\Phi$ can be represented by a tree, where the internal nodes are
NAND gates acting on their children and the leafs hold the input variables.
Here we restrict our attention to formulas that are represented by perfectly
balanced binary trees. We note that Ambainis \et al.'s algorithm can be applied
to general formulas after a proper rebalancing of the corresponding tree
\cite{Bshouty1995,Bonnet1994}. Similarly, our arguments could also be extended
to the general case.

Ambainis \et al. \cite{Ambainis2010} prove that (after efficient classical
pre-processing) $\Phi(x)$ can be evaluated with bounded-error probability using
$\sqrt{N}$ queries to $O$. The main idea is to build a weighted graph, whose
adjacency matrix, denoted as $H$, has spectrum that relates to the value of
$\Phi(x)$. Then, one simulates a discrete-time quantum walk on this graph. By
applying a phase estimation on this process for a special starting state, one
is able to infer the value of $\Phi(x)$.

Starting on the graph construction of Ambainis \et al. \cite{Ambainis2010}, we
present a different, QSVT-based approach to infer the value of $\Phi(x)$,
circumventing the quantum walk and phase estimation steps. With the
aforementioned principle of trading off lower degree polynomial approximations
by longer statistical sampling, we immediately derive an interpolating
algorithm for evaluating general NAND trees.

We present a succinct definition of $H$, referring the reader to the original
paper \cite{Ambainis2010} for a more detailed explanation. We construct a
symmetric weighted graph from the formula's tree, attaching to the root node
(call it $r$) a tail of two nodes, $r'$ and $r''$. For each node $v$, let $s_v$
be the number of variables of the subformula rooted at $v$. The weights on the
graph are defined in the following manner. If $p$ is the parent of a node $v$,
then
\begin{equation}
    \bra{v} H \ket{p} := \left( \frac{ s_v}{s_p} \right)^{1/4},
\end{equation}
with two exceptions:
\begin{enumerate}
    \item if $v$ is a leaf reading $1$, then $\bra{v} H \ket{p} := 0$
        (effectively removing the edge $(v,p)$ from the graph);
    \item $\bra{r'} H \ket{r''} := 1 / (\sqrt{2} N^{1/4})$.
\end{enumerate}

The spectrum of $H$ has the following properties \cite[Theorem~2]{Ambainis2010}:
\begin{enumerate}
    \item if $\Phi(x)=0$, then there is a zero-eigenvalue eigenstate $\ket{g}$
        of $H$ $\vert \bra{r''}\ket{g} \vert \geq 1 / \sqrt{2}$; \label{item:Phi=0}
    \item if $\Phi(x)=1$, then every eigenstate with support on $\ket{r''}$ has
        eigenvalue at least $1 / (18 \sqrt{2N})$ in absolute value.
        \label{item:Phi=1}
\end{enumerate}

That is, we can evaluate $\Phi$ by determining whether $\ket{r''}$ has a large
zero-eigenvalue component. We propose doing this within the QSVT framework.

\paragraph{Interpolation.}

The first step in our interpolation approach to evaluating NAND trees is to
construct a block-encoding of $H$. As $H$ has bounded degree and the weights of
its edges are upper bounded by $1$, we can use standard block-encoding
techniques for sparse matrices \cite{Chakraborty2018,Lin2022}. Namely, for
projectors
\begin{equation}
    \Pi, \tilde\Pi = \dyad{0^m},
\end{equation}
with $m = \bigO(\log N)$, there is a unitary $U_H$ that block-encodes $H / 3$
with $\bigO(1)$ calls to $O$. By definition, the unitary $U_H$ is such that,
for an arbitrary state $\ket{\psi}$
\begin{equation}
    U_H \ket{0^m} \ket{\psi} = \ket{0^m} \left( \frac{H}{3} \ket{\psi} \right) + \ket{\perp},
\end{equation}
where $\ket{\perp}$ is orthogonal to $\ket{0^m}$.

We would like to distinguish between the eigenstates of $H/3$ whose eigenvalue
is close to zero and those whose eigenvalue is larger than
\begin{equation}
\label{eq:delta_def}
    \frac{1}{3} \times \frac{1}{18 \sqrt{2N} } =: \delta
\end{equation}  
in absolute value. We treat this as a QSVT problem, as discussed in Section
\ref{sec:QSVT}. Indeed, let $\{\lambda_i, \ket{v_i}\}_i$ be an eigenvalue
decomposition of $H/3$ and $\ket{\psi} = \sum_i \alpha_i \ket{v_i}$ be an
arbitrary state. From Theorem \ref{thm:qsvt} and Corollary
\ref{thm:bumpfunction}, we can perform the transformation 
\begin{align}
    \ket{0^m} \ket{\psi} 
    = & \ket{0^m} \left( \sum_i \alpha_i \ket{v_i}\right) \nonumber \\
    \rightarrow & \ket{0^m} \left( \sum_i P'_{\delta, \eta, \mu}(\lambda_i)\alpha_i \ket{v_i}\right) + \ket{\perp},
    \label{eq:bumpQSVTtransformation}
\end{align}
where $P'_{\delta, \eta, \mu}$ is an approximation to the window function (as
defined in Corollary \ref{thm:bumpfunction}), with $\bigO ( (1/\delta) \log (1
/ \eta ) )$ queries to $O$.

We now have all the necessary tools to solve the problem. We start by preparing
the state $\ket{r''}$ (this does not involve any oracle queries). We then
transform $\ket{r''}$ as in \eqref{eq:bumpQSVTtransformation}. We measure the
$m$ first qubits (i.e., the block-encoding register) of the resulting state,
assigning an outcome ``yes'' if we observe $\ket{0^m}$ and an outcome ``no''
otherwise. From the spectral properties of $H$ we known that
\begin{equation}
    \mathbb{P}[\text{``yes''}] 
    \begin{cases}
    \geq  \frac{(1-\eta)^2}{2},  \text{ if } \Phi(x) = 0 \\
    \leq  \eta^2,  \text{ if } \Phi(x) = 1
    \end{cases}.
\end{equation}
So, we need to determine the bias of a Bernoulli distribution with precision no
larger than $(1- \eta)/4$. It is well-known that $\bigO(1 / (1-\eta)^2)$
samples are sufficient (and necessary) to achieve such a precision with
bounded-error probability. In summary, we can evaluate $\Phi(x)$ with
bounded-error probability by running $\bigO ( (1/\delta)
\log (1 / \eta ) )$-deep circuits
$\bigO(1 / (1-\eta)^2)$ times, amounting to a total of
\begin{equation}
    \bigO \left( \frac{1}{(1-\eta)^2} \times \frac{1}{\delta} \log\left(\frac{1}{\eta} \right) \right)
    \label{eq:NANDcomplexity}
\end{equation}
queries to $O$.

We have purposely left $\eta$ as a free parameter in our algorithm. We get the
best possible complexity by choosing $\eta = 1 -\Omega(1)$, in which case the
algorithm's query complexity is (using definition \eqref{eq:delta_def})
\begin{equation}
    \bigO\left(\frac{1}{\delta}\right) = \bigO\left( \sqrt{N}\right),
\end{equation}
recovering the scaling of Ambainis \et al. \cite{Ambainis2010}. But this choice
of $\eta$ requires running circuits of depth also in $\bigO(\sqrt{N})$. Suppose
now that we want to limit the circuit depth to some maximum value $D$. We can
run the same algorithm, setting this time $\eta$ to be
\begin{equation}
    \eta = \bigO\left(2^{-\delta D}\right).
\end{equation}

Replacing into expression \eqref{eq:NANDcomplexity}, we find that
\begin{equation}
    Q(\Phi; D)  = \bigO\left( \frac{ N}{D}  + \sqrt{N}\right).
\end{equation}

\paragraph{Parallelization.}

The problem of evaluating NAND trees is also amenable to parallelization. The
key observation is that, if for any given level of the tree we know the logical
value of all the nodes at that level, then we can infer $\Phi(x)$ without
performing any more queries to the input. Therefore, we solve the problem if,
for every node $v$ at that level, we run the quantum algorithm for evaluating
the NAND tree rooted at $v$.

Say that we want to limit our circuit depths to $D$. We partition the input
variables into $\bigO( N/D^2)$ subsets of $\bigO(D^2)$ variables each. To each
subset of variables corresponds a subtree of the total tree. For each such
subtree, we evaluate the logical value of the root node with an error
probability bounded by $D^2/N$, which we can do with $\bigO(\sqrt{D^2} \log(N/D^2))$
queries to $O$. Since we repeat this for all subtrees, the hybrid query
complexity becomes
\begin{equation}
    Q(\Phi; D) =
    \bigO\left(\frac{N}{D} \log\left(\frac{N}{D}\right)+ \sqrt{N} \right).
\end{equation}

We find that both the interpolation and parallelization methods can be applied
for evaluating balanced binary NAND trees. Although the resulting complexities
are close, the parallelization approach comes with an extra $\log(N/D)$ factor.
This problem illustrates that there are also situations where interpolation is
advantageous over parallelization.

\section{Conclusions}

In this paper, we suggest two distinct approaches for adapting a quantum
algorithm to a restricted-depth setting: parallelization and interpolation. An
algorithm is said to be ``parallelizable''  whenever we can split its action
into smaller, independent sub-problems; and ``interpolatable'' if the loss of
information caused by shortening the circuit depth can be compensated by
repeated runs of the shorter circuit. Therefore, informally, these two methods
can be understood as either ``breaking up the input'' (for parallelization) or
``breaking up the unitary procedure'' (for interpolation).

We argue that Quantum Singular Value Transformations (QSVT) closely relate to
the notion of interpolation, rather than parallelization. For QSVTs, a smaller
circuit depth corresponds to a polynomial approximation to a target function
of lower degree, which needs to be compensated by longer statistical sampling.

We apply these approaches to two problems with known quantum speed-ups: the
\(k\)-threshold function and perfectly balanced NAND trees. To the best of our
knowledge, neither of these problems had been studied in a hybrid,
restricted-depth setting. For the  \(k\)-threshold function, we show that
parallelization offers the best performance by a factor of $\tbigO(k)$ (in
terms of query complexity). In contrast, for evaluating perfectly balanced NAND
trees the interpolation method is the most efficient, differing by a factor of
$\bigO(\log (N/D))$. This way, we demonstrate that no technique
(parallelization or interpolation) is strictly better than the other -- each
one may be the best option depending on the problem at hand.

This shows that, when designing a quantum-classical hybrid algorithm obeying
certain (query) depth limitations, both of the proposed techniques can be
explored as a strategy for maintaining some of the speedup (over a fully
classical approach) of a quantum unrestricted-depth counterpart. Furthermore,
given the close connection between (depth unrestricted) algorithms formulated
in terms of QSVTs and the interpolation method, this implies that, when
searching for hybrid quantum-classical algorithms for a particular problem, it may
be a good option to start by formulating a (depth unrestricted) QSVT algorithm
for the problem, and then seeking to interpolate it.

We note that we only offered an example of a problem (perfectly balanced NAND
trees) where the interpolation beats parallelization by a logarithmic factor.
It would be interesting to find a problem for which the interpolation procedure
is polynomially more efficient than the corresponding parallelization, to rule
out the possibility that parallelization, whenever applicable, is always
optimal up to logarithmic factors. We leave the existence of such a problem as
an open question.

The definitions we have provided for the terms ``parallelization'' and
``interpolation'' are not strictly rigorous; they should be seen as general
strategies for restricted-depth computing, rather than formal notions. This
does not preclude that in some situations the two strategies may be
simultaneously at play, or that these classifications may not apply. As such,
we expect there is room for discussion on what other classes of methods may
exist besides the ones discussed here, and for other systematic approaches to
hybridization.

\begin{acknowledgments}
    We thank ~R.~de~Wolf for his comments on quantum query lower bounds for
    the problem of quantum counting, in the context of calculating the
    threshold function, and N.~Stamatopoulos for his comments regarding the
    proof of the parallelization method for the threshold function. We also
    thank the support from FCT -- Funda\c{c}\~{a}o para a Ci\^{e}ncia e a
    Tecnologia (Portugal), namely through project UIDB/04540/2020, as well as
    from projects QuantHEP and HQCC supported by the EU QuantERA ERA-NET Cofund
    in Quantum Technologies and by FCT (QuantERA/0001/2019 and
    QuantERA/004/2021, respectively), and from the EU Horizon Europe Quantum
    Flagship project EuRyQa (101070144). DM and MM acknowledge the support from
    FCT through scholarships 2020.04677.BD and 2021.05528.BD, respectively.
\end{acknowledgments}

\bibliography{citations}

\vfill

\appendix

\section{Threshold function -- proof of parallelization method\label{sec:threshold_proof}}

From Brassard \et al. \cite{Brassard2002}, there is a constant $c$ such that
the error for our estimate of $\HammingWeight{x \Sslash V_i}$ is bounded as
\begin{equation}
    \epsilon_i < c \frac{\sqrt{N \HammingWeight{x \Sslash V_i} / p}}{D}.
\end{equation}

We analyse separately the cases where $\Threshold_k(x) = 0$ and
$\Threshold_k(x) = 1$.

If $\Threshold_k(x) = 0$, the following possible relations between $p$, $k$,
and $\HammingWeight{x}$ need to be considered.

\begin{enumerate}
\item $p \log p \leq \HammingWeight{x} \leq k$.
    From the result of Raab and Steger (equation \eqref{eq:ballsintobins}), we
        know that $\HammingWeight{x \Sslash V_i} = \bigO(\HammingWeight{x} /
        p)$ for all $i$. So, by our expression for the error
        \eqref{eq:xprecision}, we see that $\epsilon_i =
        \bigO(\sqrt{\HammingWeight{x} / k}) = \bigO(1)$.

\item $\HammingWeight{x} \leq p \log p \leq k$.
    Now we know that $\HammingWeight{x \Sslash V_i} = \bigO(\log p)$. So,
        for all $i$, $\epsilon_i = \bigO(\sqrt{p \log p / k}) = \bigO(1)$.

\item $\HammingWeight{x} \leq k \leq p \log p$.
    Equation \eqref{eq:ballsintobins} ensures that $\HammingWeight{x \Sslash
        V_i} = \bigO(\log p)$. From the expression for the error, we see that
        $\epsilon_i = \bigO(\sqrt{\log p / \log p}) = \bigO(1)$.
\end{enumerate}

That is, there is a choice of constants that guarantees that $\epsilon_i < 1/2$
for all $i$ with bounded probability. In that case, we estimate each
$\HammingWeight{x \Sslash V_i}$ exactly, and so we exactly infer
$\HammingWeight{x}$ and consequently the value of $\Threshold_k(x)$.

If  $\Threshold_k(x) = 1$, the proof is slightly different. Again, we consider
three scenarios.

\begin{enumerate}
    \item $p \log p \leq k \leq \HammingWeight{x}$.
        Equation \eqref{eq:ballsintobins} tells us that $\HammingWeight{x
        \Sslash V_i} = \bigO(\HammingWeight{x} / p)$. Combining with
        \eqref{eq:xprecision} we see that there is a (controllable) constant
        $C$ for which
        \begin{equation}
            \epsilon_i < C \sqrt{\frac{\HammingWeight{x}}{k}}
        \label{eq:gettablebound}
        \end{equation}
        for all $\HammingWeight{x}, k$. Unlike before, we cannot guarantee that
        $\epsilon_i$ is kept below $1/2$ for all $\HammingWeight{x}$. But we
        can make sure that our estimate for the Hamming weight is always
        greater than $k$. Let \(X_j\) be the random variables corresponding to
        the estimations of each $\HammingWeight{x \Sslash V_j}$, and
        $\sigma_j^2$ the corresponding variances. From the Chebyshev bound,
        \begin{multline}
            \Probability\biggl[\bigl\vert{\sum_j X_j - \HammingWeight{x}}\bigr\vert > \HammingWeight{x} - k\biggr] <
            \abs{\frac{\sum_j \sigma_j}{\HammingWeight{x} - k}}^2 < {} \\
            {} < \abs{C' \, p\, \frac{\sqrt{\HammingWeight{x}/k}}{\HammingWeight{x}-k}}^2
        \end{multline}
        for some constant $C'$. Thus we can attain with constant probability an
        estimation of $\HammingWeight{x}$ with error within
        $\HammingWeight{x}-k$ if there exists a constant $C'$ such that there
        exists a value $\HammingWeight{x}^*$ satisfying:
        \begin{itemize}
            \item If $\HammingWeight{x}<\HammingWeight{x}^*$, the error in the
                estimation of each $\HammingWeight{x \Sslash V_j}$ is less than
                $1/2$, such that the estimate of $\HammingWeight{x}$ is exact,
            \item If $\HammingWeight{x}>\HammingWeight{x}^*$, then $C'
                \sqrt{\HammingWeight{x}/k} < (|x|-k)/p$, bounding the error
                probability to be constant. 
        \end{itemize}

        Choosing $\HammingWeight{x}^* = k + p \log p$, one can check that $C =
        1 / 4(C')$ satisfies the conditions above.
 
    \item $k \leq p \log p \leq \HammingWeight{x}$.
        Again, for all $i$, $\HammingWeight{x \Sslash V_i} =
        \bigO(\HammingWeight{x} / p)$. Combining this with the expression for
        the error \eqref{eq:xprecision}, we get $\epsilon_i =
        \bigO(\HammingWeight{x} / p \log p)$. The proof follows the same steps
        as the ``$p \log p \leq k \leq \HammingWeight{x}$'' case.

    \item $k \leq \HammingWeight{x} \leq p \log p$.
        From equation \eqref{eq:ballsintobins} we known that $\HammingWeight{x
        \Sslash V_i} = \bigO(\log p)$ for all $i$. Then, $\epsilon_i =
        \bigO(\sqrt{\log p / \log p}) = \bigO(1)$. So, in this case we can also
        ensure that we estimate $\sum_i \HammingWeight{x \Sslash V_i}$ exactly.
\end{enumerate}

\section{Total Non-Constant Symmetric Boolean Functions \label{sec:symmfunctions}}

We have shown, before, how to interpolate the \(k\)-threshold function based on
Quantum Singular Value Transformations. A similar interpolation scheme to the
one we have shown can actually be applied to the calculation of any symmetric
boolean function, as we now show. Furthermore, we show that a similar
difference exists between the scaling for this interpolation and the scaling
for a parallelization procedure.

We start by reviewing an intermediate claim of Beals \et al.~\cite{Beals2001}:

\begin{lemma}(Part of Theorem 4.10 of Beals \et al.\ \cite{Beals2001})
    \label{lemma:beals_threshold}
    For a symmetric boolean function \(f\), if given an algorithm that outputs
    \(\HammingWeight{X}\) if \(\HammingWeight{X} < (N - \Gamma(f))/2\) or
    outputs ``in'' otherwise, with \(Q\) queries to the oracle, immediately
    there is an algorithm that computes \(f\) with \(Q\) queries to the oracle.
\end{lemma}
\begin{proof}
    Let \Algorithm{A} be an algorithm as outlined in the lemma, requiring \(Q\)
    queries to the oracle. By definition of \(\Gamma(f)\), \(f\) is constant for
    \(X\) such that \(\HammingWeight{X} \in [(N-\Gamma(f))/2, (N+\Gamma(f))/2]\).
    Therefore, let \Algorithm{A'} be an algorithm that runs \Algorithm{A}, and
    then:
    \begin{itemize}
        \item If \Algorithm{A} outputs ``in'', \Algorithm{A'} outputs
            \(f((N-\Gamma(f))/2)\),
        \item If \Algorithm{A} outputs \(\HammingWeight{X}\), \Algorithm{A'}
            outputs \(f(\HammingWeight{X})\).
    \end{itemize}
    
    \Algorithm{A'} requires only as many queries as \Algorithm{A}.
\end{proof}

Now, departing from Beals \et al.'s proof, we rephrase the construction of an
algorithm matching the description of \cref{lemma:beals_threshold} in terms of
Quantum Singular Value Transformations.

We start with the following lemma of Low and Chuang \cite{Low2017}:

\begin{lemma}{\cite{Low2017}}
    \label{lemma:qsvt_erf}
    For a given \(k\in\Reals\), \(\delta \in [-1, 1]\) and \(\epsilon\in(0,
    \bigO(1))\), there exists a real polynomial \(p(x)\) satisfying
    \begin{gather*}
        \abs{p(x)} \leq 1 \Where x \in [-1, 1], \text{ and} \\
        \abs{p(x) - \erf(k (x - \delta))} \leq \epsilon \Where x \in [-1, 1].
    \end{gather*}
    with polynomial degree
    \begin{equation}\label{eq:qsvt_erf}
        \deg(p) = \bigO*(\sqrt{\qty(\log\frac{1}{\epsilon}) \qty(k^2 + \log\frac{1}{\epsilon})}).
    \end{equation}
\end{lemma}

From this lemma follows the already mentioned construction for a polynomial
approximation to the threshold function, which we restate:

\begin{corollary}{\cite{Low2017}}
    \label{corollary:qsvt_threshold}
    For a given \(\delta \in [-1, 1]\), \(\epsilon \in (0, \bigO(1))\), \(\eta
    \in (0, 1/4)\), there exists a real polynomial \(p\) satisfying
    \begin{align*}
        \abs{p(x)} \leq 1 & \Where x \in [-1, 1] \\
        \abs{p(x) - 1} \leq \eta & \Where x \in [-1, \delta - \epsilon], \\
        \abs{p(x)} \leq \eta & \Where x \in [\delta + \epsilon, 1],
    \end{align*}
    and with polynomial degree
    \[ \deg(p) = \bigO*(\frac{1}{\epsilon} \log \frac{1}{\eta}). \]
\end{corollary}

To make use of these polynomial transformations, we also recall the block
encoding of the quantities of interest, which are the same as for the
\(k\)-threshold case; for unitary
\[ U = \quad\raisebox{1.2em}{%
    \Qcircuit @C=1.1em @R=.5em {
        & \qw{{}^{n}/} & \gate{H^{\otimes}}     & \multigate{1}{O_X} & \qw \\
        & \qw          & \qw                    & \ghost{O_X}        & \gate{X} & \qw \\
    }
}
\]
and \(\Pi = \dyad*{0^{n+1}}\), we have that \((\Identity_{2^n} \otimes
\dyad{1}) U \Pi\) is a block encoding of \(\sqrt{\HammingWeight{X}/N}\), and
\((\Identity_{2^n} \otimes \dyad{0}) U \Pi\) is a block encoding of
\(\sqrt{(N-\HammingWeight{X})/N}\).

Now we first determine if the Hamming weight of the input should produce output
``in'' or not, which is, essentially, the task of calculating the
\(k\)-threshold function with \(k = (N - \Gamma(f))/2\) and with \(k' = (N +
\Gamma(f))/2\). As stated in the body text, the case of threshold \(k'\) can
be reduced to the case of threshold \(k\) calculated for the complement of the
Hamming weight \(N - \HammingWeight{X}\), and so we conclude that this step
requires \(\bigO(2\sqrt{N(N-\Gamma(f))})=\bigO(\sqrt{N(N-\Gamma(f))})\)
applications of the oracle.

In the event that we find \(\HammingWeight{X}\) to be smaller than
\((N-\Gamma(f))/2\), or larger than \((N+\Gamma(f))/2\), it remains to output
the Hamming weight of \(X\), or of \(\neg{X}\), respectively. We consider
henceforth the case of \(\HammingWeight{X} < (N-\Gamma(f))/2\), from which 
generalization is easy.

Note first that performing bisections on \(\HammingWeight{X}\) for
\(\HammingWeight{X} \in [0, (N-\Gamma(f))/2]\) corresponds to performing
successive threshold operations for thresholds \(k' < (N-\Gamma(f))/2\), so, by
binary search, we have that we may find \(\HammingWeight{X}\) with
\(\bigO[\sqrt{N(N-\Gamma(f))} \log^2(N-\Gamma(f))]\) applications of the
oracle, where one of the \(\log\) factors is due to the binary search, and the
other to error probability bounding. However, by making direct use of
\cref{lemma:qsvt_erf}, the \(\log\) factors can be significantly lowered.
Consider the following lemma:

\begin{lemma} \label{lemma:cut_in_half}
    Given the block encoding of a value \(z \in [a,b] \subseteq [-1, 1]\), it
    is possible to determine \([a', b'] \subseteq [a,b]\) such that \(z \in
    [a', b']\), and \((b' - a') \leq (b - a)/2\), with
    \begin{equation}
        D_\text{round} = \bigO*(\frac{1}{b - a})
    \end{equation}
    coherent applications of the oracle, and
    \begin{equation}
        T_\text{round} = \bigO*(\frac{1}{b-a} \log \frac{1}{E})
    \end{equation}
    total applications of the oracle, with probability of error at most \(E\).
\end{lemma}
\begin{proof}
    Fix \(\eta \in O(1)\), for example, \(\eta = 1/8\). Using QSVT, create a
    block encoding of \(P(z)\), where \(P\) is the polynomial approximating
    \(\erf(k(x-\delta))\) up to absolute error \(\epsilon\) (to be determined),
    with \(k = \frac{2}{b-a} \erf\inv(1-2\eta)\) and \(\mu = (b-a)/2\). For a
    choice of \(\sigma\), after \(\bigO(\log(1/E)/\sigma^2)\) samples of this
    encoding, one obtains an estimate for \(\erf(k(z-\mu))^2\) up to precision
    \(\sigma + \epsilon^2\) with error probability \(E\); denote this estimate
    \(\tilde p\). This estimate implies a new window \([a', b']\) for \(z\)
    satisfying
    \begin{multline}
        b' - a' \leq \frac{1}{k}
            \big(\erf\inv(2 \sqrt{\tilde p} - 1 + 2(\sigma + \epsilon)) - {} \\
                \erf\inv(2\sqrt{\tilde p} - 1 - 2(\sigma + \epsilon)) \big)
    \end{multline}
    which in turn satisfies, with our choice of \(k\),
    \begin{multline}
        \frac{b'-a'}{b-a} \leq \frac{2}{\erf\inv(1-2\eta)} (\sigma+\epsilon) \times {} \\
        {} \times \max_{y\in[-2,2]} (\erf\inv)' (2\sqrt{\tilde{p}}-1 + y(\sigma+\epsilon)).
    \end{multline}
    
    Demanding that
    \begin{equation} \label{eq:constraint}
        \sigma + \epsilon \Please\leq \eta/4,
    \end{equation}
    we have
    \begin{equation}
        \frac{b'-a'}{b-a} \leq \frac{\eta}{4} \frac{\sqrt{\pi}}{\erf\inv(1-2\eta)} e^{(\erf\inv(1-\eta))^2}
    \end{equation}
    which, for \(\eta = 1/8\), has a right-hand side of less than one half.
    
    Since \(\eta \in \bigO(1)\), we may choose \(\sigma \in \bigO(1)\) and
    \(\epsilon \in \bigO(1)\) satisfying the constraint \eqref{eq:constraint},
    and thus follows that this procedure requires
    \begin{equation}
        M_\text{round} = \bigO*(\frac{1}{\sigma^2} \log\frac{1}{E}) = \bigO*(\log\frac{1}{E})
    \end{equation}
    measurements of a circuit encoding a polynomial transformation of degree
    (cf.\ equation \eqref{eq:qsvt_erf})
    \begin{equation}
        D_\text{round} = \bigO*(k) = \bigO*(\frac{1}{b-a})
    \end{equation}
    or, equivalently, the same number of coherent queries. The total number of
    queries is therefore
    \begin{equation}
        T_\text{round} = D_\text{round} \cdot M_\text{round} = \bigO*(\frac{\log E\inv}{b-a})
    \end{equation}
    as claimed.
\end{proof}

By repeating the procedure given in the lemma above, we may reduce the window
for \(\sqrt{\HammingWeight{X}/N}\) until this value is unambiguous. This
requires a final window of size \(\Delta =
\frac{1}{2}(\sqrt{N-\Gamma(f)}-\sqrt{N-\Gamma(f)-1})\), which in turn requires
\(\log\/(\sqrt{N-\Gamma(f)}/\Delta) = \bigO[\log\/\{\sqrt{N}(N - \Gamma(f))\}]\) rounds of
application of the lemma. 

Since we wish any of these rounds to fail with probability at most \(1/3\),
this requires that each round fails with probability at most
\(1/(3\log[\sqrt{N}(N-\Gamma(f))])\).

Therefore, overall, this procedure requires
\begin{equation}
    D = \bigO*(\sqrt{N(N - \Gamma(f))})
\end{equation}
maximum coherent oracle calls, and a total number of oracle calls
\begin{equation}
    T = \bigO*(\sqrt{N - \Gamma(f)} \log\log[\sqrt{N}(N-\Gamma(f))]).
\end{equation}

Performing the interpolation now is straightforward: instead of demanding the
procedure from lemma \ref{lemma:cut_in_half} be repeated until the window is so
small that the Hamming weight of \(X\) is unambiguous, we instead choose a final
window size \(\Delta'\) that respects the given coherent query limit. After this
limit has been reached, we ``switch'' to statistical sampling until the
final window size for the value of \(\sqrt{\HammingWeight{X}/N}\) is
\(\Delta\). This procedure therefore is split into two steps; following an
analysis analogous to the one for the unbound case, we conclude that for some
choice of \(\Delta'\), the first phase requires maximum coherent query depth
\begin{equation}
    D_\text{first} = \bigO*(\frac{1}{\Delta'})
\end{equation}
and total query count
\begin{equation}
    T_\text{first} = \bigO*(\frac{\log E\inv}{\Delta'}).
\end{equation}

Using the fact that \((\erf\inv)'(x) \geq \sqrt{\pi}/2\), one may then conclude
that and the second phase requires corresponding
\begin{gather}
    D_\text{second} = \bigO*(\frac{1}{\Delta'} \sqrt{\log(C\frac{\Delta'}{\Delta})}) \\
    T_\text{second} = \bigO*(\frac{\Delta'}{\Delta^2} \sqrt{\log(C\frac{\Delta'}{\Delta})}\log E\inv)
\end{gather}
where, again, \(E\) is the error probability, and \(C\) is a constant in
\(\bigO(1)\).

Choosing this intermediate window size \(\Delta'\) to be \(\Delta^{1-\alpha}\),
for \(\alpha \in [0,1]\), we recover complexities analogous to those verified
for \(\alpha\)-Quantum Phase Estimation \cite{GiurgicaTiron2022,Magano2022}:
\begin{multline}
    D(\alpha) = D_\text{first} + D_\text{second} = {} \\
    {} = \bigO*(\Delta^{\alpha-1} \sqrt{\log(C\Delta^{-\alpha})})
\end{multline}
\begin{multline}
        T(\alpha) = T_\text{first} + T_\text{second} = {}\\
        {}=\bigO*(\Delta^{-(1+\alpha)} \sqrt{\log(C\Delta^{-\alpha})} \log E\inv)
\end{multline}

Using again the fact that \(\Delta\inv = \bigO(\sqrt{N(N-\Gamma(f))})\), and
with considerations as to the error probability identical to before, we finally
have
\begin{gather}
    D(\alpha) = \tbigO*({[ N(N-\Gamma(f)) ]}^{(1-\alpha)/2}) \\
    T(\alpha) = \tbigO*({[N(N-\Gamma(f))]}^{(1+\alpha)/2}).
\end{gather}

\end{document}